\documentclass[12pt]{article}
\usepackage[utf8]{inputenc}

\usepackage{amsmath,amsfonts,amsthm,amssymb}
\usepackage{setspace}
\usepackage{fancyhdr}
\usepackage{lastpage}
\usepackage{extramarks}
\usepackage{chngpage}
\usepackage{soul,color}
\usepackage{graphicx,float,wrapfig}
\usepackage{listings}
\usepackage{tikz}
\usepackage{harpoon}
\usepackage[ruled]{algorithm2e}
\usepackage{hyperref}
\usepackage{bm}
\usepackage{indentfirst}
\usepackage[round]{natbib}
\usepackage{verbatim}
\usepackage{subfigure}
\usepackage[moderate]{savetrees}
\usepackage{times}
\usetikzlibrary{decorations.pathreplacing}

\hypersetup{colorlinks=true,
            citecolor=blue,
            linkcolor=blue,
            urlcolor=blue,
			bookmarksopen=true,
			pdfstartview={XYZ null null 1.00},
			pdfpagelayout={SinglePage}
}

\usepackage{paralist}

\newtheorem{theorem}{Theorem}
\newtheorem{lemma}{Lemma}
\newtheorem{proposition}{Proposition}

\newtheorem{definition}{Definition}

\newtheorem{example}{Example}

\newtheorem{remark}{Remark}

\usetikzlibrary{arrows,graphs}
\newcommand{\enabstractname}{Abstract}
\newcommand{\cnabstractname}{摘要}
\newenvironment{enabstract}{%
  \noindent\mbox{}\hfill{\bfseries \enabstractname}\hfill\mbox{}\par
  \vskip 2.5ex}{\par\vskip 2.5ex}

\graphicspath{{figure/}}                                 
\usepackage[a4paper]{geometry}                           
\begin{document}
\title{Strategy-proof Market Segmentation against \\ Price Discrimination\thanks{
First Version June 1, 2022. An earlier version of this paper was circulated under the title ``Stable Market Segmentation against Price Discrimination''. We are especially grateful to Dirk Bergemann, B\r{a}rd Hastard, Asher Wolinsky, and Jidong Zhou for the detailed comments, which helped us improve the manuscript. We thank Federico Echenique, Nima Haghpanah, Zhiguo He, Shota Ichihashi, Jin Li, Qingmin Liu, Ron Siegel, Kai Hao Yang, Junjie Zhou, and the audiences at various seminars and conferences for their helpful comments and suggestions. This paper is supported by the National Natural Science Foundation of China (No. 72473137, 72192801, 72103190), Beijing Natural Science Foundation (Z220001), and the ``fund for building world-class universities (disciplines) of Renmin University of China”. Any remaining errors are our own. \textit{Email}: Kuang (kuang@ruc.edu.cn); Li (sanxi@ruc.edu.cn); Liu (yi.liu.yl2859@yale.edu); Yu (yangyu@cup.edu.cn).}
}
\author{
    \textsc{Zhonghong Kuang} \\ \textit{Renmin U}
    \and \textsc{Sanxi Li} \\ \textit{Renmin U}
	\and \textsc{Yi Liu} \\ \textit{Yale}
	\and \textsc{Yang Yu} \\ \textit{CUP}
}
\maketitle

\begin{enabstract}
\noindent Data regulations increasingly enable consumers to switch among market segments, making segmentation an endogenous outcome of strategic interaction. We study a model in which consumers choose segments before a monopolist sets segment-specific prices, and define a segmentation as strategy-proof when no consumer with positive measure can profitably deviate. Our main result provides a complete welfare characterization: in every strategy-proof segmentation, producer surplus is pinned at the uniform monopoly profit, consumer surplus ranges from the uniform monopoly level to the buyer-optimal level, and every consumer is weakly better off. We construct strategy-proof segmentations attaining every feasible outcome in this range. A finite-consumer model microfounds our solution concept, with equilibrium outcomes converging to our characterization as the population grows large.

\vspace{1ex}

\end{enabstract}

\newpage

\section{Introduction}

Online sellers can leverage consumer data for market segmentation and price discrimination, potentially harming consumer welfare. A straightforward remedy, widely discussed by scholars and competition authorities, is to prohibit price discrimination altogether. Recent privacy protection advancements such as the GDPR offer an alternative approach: enabling consumers to strategically switch among market segments, which can enhance consumer surplus. A prerequisite for such mobility, \emph{the Right to be Forgotten}, is already legally mandated.\footnote{All iOS apps on the App Store are required to offer an in-app account deletion option by July 30, 2022 (see \url{https://developer.apple.com/support/offering-account-deletion-in-your-app}). Moreover, all iOS apps should let people use them without a login (see App Store Review Guideline 5.1.1(v) in \url{https://developer.apple.com/app-store/review/guidelines/\#data-collection-and-storage}).} Yet consumers are decentralized, self-interested, and individually insignificant. \emph{Can they nonetheless be effectively protected against a monopolist who engages in price discrimination?} To explore this question, we study an extreme scenario in which consumers face no costs when switching market segments, thereby providing a non-cooperative, decentralized foundation for market segmentation.

In our model, a monopolistic producer sells a homogeneous product to a unit-demand consumer population. Production exhibits constant marginal cost, normalized to zero, and each consumer's reservation price takes a value from a finite set. The timing is as follows. In the first stage, each consumer simultaneously selects a market segment to join. In the second stage, the producer observes the value distribution within each segment and sets a take-it-or-leave-it price for each. Finally, each consumer makes a purchase decision.

A key modeling challenge is to capture consumers who are both numerous and individually non-negligible. We address this by representing the consumer population as a continuum while defining a \emph{profitable unilateral deviation} as a shift by a consumer with infinitesimal positive measure to a different market segment that yields a strictly higher payoff. Crucially, such a deviation can alter prices: online sellers typically employ algorithmic pricing (see, e.g., \cite*{johnson2023platform} and \cite{shota2026algorithmic}), and consumers rationally anticipate this. We define a market segmentation as \emph{strategy-proof} if it satisfies two conditions: (i) the producer sets an optimal price in every market segment, including off-path ones, and (ii) no profitable unilateral deviation exists for any consumer.

We restrict the seller to third-degree price discrimination: a uniform price is set within each market segment. This restriction admits two interpretations. First, the seller may lack the granular data needed to discriminate at the individual level, observing only the type distribution within each segment. Second, even when individual valuations are observable, regulatory or legal constraints may prohibit perfect price discrimination. For instance, legal provisions in New York and California bar the use of certain personal characteristics (e.g., gender) for pricing purposes, even when these traits are correlated with willingness to pay. A further example is Steam, where product prices vary across regions based on users' IP addresses but are uniform within each region. Consumers can exploit this structure by using virtual private networks (VPNs) to access prices offered in lower-income countries. Here, technical constraints reinforce the restriction: producers can aggregate data only at the IP-address level, making individual-level price discrimination infeasible.

Our main result characterizes the welfare consequences of all strategy-proof market segmentations. When consumers strategically choose market segments before the producer sets prices, producer surplus is always fixed at the uniform monopoly profit, while consumer surplus ranges from the uniform monopoly level to the maximum attainable under arbitrary segmentation (the \emph{buyer-optimal} outcome). Moreover, no consumer faces a price above the uniform monopoly price, so all consumers are weakly better off.

We establish attainability by constructing strategy-proof segmentations for both extreme outcomes. For the uniform monopoly outcome, the argument is immediate: absent any segmentation, the producer charges the uniform monopoly price, and any deviation would require creating a new market, which is not profitable for any consumer. For the buyer-optimal outcome, we partition the aggregate market into \emph{extremal markets} in which every valuation in the support is an optimal price. The producer therefore charges the minimum valuation in each market.\footnote{Consider all markets whose optimal price equals the lowest willingness to pay. Since no deviation can lead to a lower price, consumers in these markets have no incentive to leave. Iterating this argument from the lowest-priced markets to the highest confirms that the construction is strategy-proof.} Prices in extremal markets are sensitive to entry, which is precisely what enables a strategy-proof segmentation to sustain different prices across markets.

Having established that both extreme outcomes are attainable, we construct a family of strategy-proof segmentations implementing every intermediate welfare outcome. A naive approach---taking convex combinations of strategy-proof segmentations---does not work, as convex combinations need not preserve strategy-proofness. Our construction instead starts from the buyer-optimal segmentation and gradually merges the highest-priced markets into a single market whose share is governed by a tuning parameter ranging from zero to one. At one extreme, the construction coincides with the buyer-optimal outcome; at the other, it yields the uniform monopoly outcome. Since consumer surplus varies continuously in this parameter, provided the critical market can be proportionally divided, every intermediate level of consumer surplus is attained.

Two natural alternatives exist for modeling the decentralized formation of market segmentation in a continuum economy. Under the first, the seller's pricing strategy in each segment depends only on the type distribution, rendering any individual deviation a measure-zero event with no price impact. This specification is too restrictive: all segments must charge the uniform monopoly price, so only the monopoly outcome is attainable. Under the second---subgame perfect equilibrium (SPE)---even a measure-zero deviation can alter prices. This specification is too permissive: the set of welfare outcomes under SPE generally contains that under strategy-proof market segmentation.

Our solution concept occupies the middle ground between these two extremes. To justify this position, we study a discrete-consumer model with a large but finite number of consumers, where each player has non-negligible mass. We characterize the Nash equilibrium of this game given the monopolist's optimal pricing and show that, as the number of consumers grows to infinity, the set of equilibrium welfare outcomes converges to precisely the set identified by our notion of strategy-proofness. This convergence result provides a micro-foundation for our solution concept as the appropriate formalization of decentralized consumer behavior.

Our work springs from two strands of literature. First, our paper primarily belongs to the literature studying the welfare consequences of price discrimination (e.g., \cite*{aguirre2010monopoly} and \cite{cowan2016welfare}).
In a seminal paper, \cite*{bergemann2015limits} characterize all possible welfare consequences in third-degree price discrimination with exogenous segmentation. The shaded triangle in \autoref{fig:sur_tri} depicts all available surplus pairs, in which $A$ marks the uniform monopoly outcome and $C$ marks the buyer-optimal segmentation.\footnote{Point $B$ marks the first-degree price discrimination outcome, in which consumers with the same valuations are grouped. Point $D$ marks the outcome where social welfare is minimized.} The line segment between $A$ and $C$ represents the set of possible welfare pairs in our results. \cite{ichihashi2023data} investigate how the surplus triangle shrinks when the producer has private information unknown to the central designer.

\begin{figure}[!htb]
\centering
\begin{minipage}[t]{0.45\linewidth}
\begin{tikzpicture}[scale=5]
\draw[line width=1.5pt](0,0)--(0.8,0)node[below]at(0.5,0){Consumer surplus ($u$)};
\draw[line width=1.5pt](0,0)--(0,1)node[left]at(0,0.5){\rotatebox{90}{Producer surplus ($\pi$)}};
\draw[line width=1.5pt](0,0.2)--(0,0.9)--(0.7,0.2)--(0,0.2);
\fill[fill opacity=0.2](0,0.2)--(0,0.9)--(0.7,0.2);
\fill (0,0.2) circle (0.02)node[above]at(0.05,0.2){$D$};
\fill[color = red] (0.7,0.2) circle (0.02)node[above]{$C$};
\fill[color = red] (0.35,0.2) circle (0.02)node[above]{$A$};
\fill (0,0.9) circle (0.02)node[above]at(0.05,0.9){$B$};
\draw[line width=1.5pt,color=red] (0.35,0.2) -- (0.7,0.2);
\end{tikzpicture}
\end{minipage}
\caption{Surplus Triangle}\label{fig:sur_tri}
\end{figure}
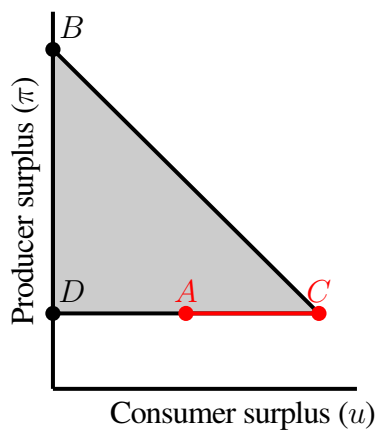

The market segmentation problem analyzed in \cite*{bergemann2015limits} can also be viewed as an information design problem with a single consumer. Then, \cite{roesler2017buyer} study how an uninformed buyer, with a known prior value distribution, learns the distribution. \cite{condorelli2020information} further endogenize the selection of value distribution. These studies mostly focus on settings with a single buyer and a homogeneous product, which is not suitable for endogenizing the market segmentation. Hence, we incorporate a large number of strategic consumers, which complements previous studies.\footnote{See \cite{ichihashi2020online,deb2021multi,haghpanah2019pareto,haghpanah2022The} for examples of bilateral trade models with multiple products and \cite*{armstrong2022consumer,belleflamme2020competitive,elliott2019market,chen2020Competitive,rhodes2023personalized} for examples with multiple strategic and competitive producers.}

Second, this paper is also linked to the literature on strategic consumers. In a cooperative setting, \cite{haghpanah2023A} study the formation of market segmentation by consumers, whereas our paper considers the formation as the equilibrium of a non-cooperative game. As a result, the markets they analyze must be efficient, the producer surplus may arise, and the outcome may not be Pareto-improving compared with the aggregate market. In contrast, the market segmentations analyzed in our paper may be inefficient, the producer surplus is fixed, and the outcome is always Pareto-improving. 

One mainstream of literature assumes repeated purchases from a monopolist. In such \emph{behavior-based price discrimination} scenarios, the monopolist obtains information about consumers through their purchase history, and consumers act myopically or strategically against the monopolist \citep{fudenberg2000customer,gehrig2007information,shen2018behavior,bonatti2020consumer}.\footnote{See the detailed survey conducted by \cite{fudenberg2006behavior}.}
\cite{acquisti2005conditioning} argues that purchase history may harm consumers in later periods, which, in turn, forces consumers to protect their privacy.
\cite{chen2009dynamic} consider strategic consumers seeking the lowest price when facing dynamic pricing. Our paper differs structurally from those studies, as strategic consumers in our model make only a single purchase decision, while their strategic feature arises from their interaction with other consumers.

The strategic concerns of consumers are also present in their voluntary disclosure decisions before the purchase. For example, \cite*{ali2020voluntary} investigate the incentive to voluntarily disclose hard evidence (exact information or a range about personal preferences), which builds the market segmentation. \cite{hidir2021privacy} study the soft information about consumers' valuations and introduce the incentive-compatible market segmentation. 
In their settings, the strategic side of consumers is manifested by their decisions on their data, which indirectly influences market segmentation. While in ours, each consumer directly chooses the market segment.

The remainder of the paper proceeds as follows. \autoref{Section:Model} introduces the model. \autoref{Section:Result} presents our main result and its proof. \autoref{Section:Discussion} further discusses our solution concept and justifies it using a discrete-consumer model. Some technical proofs are collected in the Appendix.

\section{Model}\label{Section:Model}

A monopolistic producer sells a product to $N$ consumers, with a constant marginal cost normalized to zero. Since $N$ should be very large in the real world, we can approximate that the seller serves a continuum of consumers for convenience. Consumers are uniformly distributed within $[0,1]$, each with unit demand. Consumers' willingness-to-pay of the product is determined by a function $V:[0,1]\rightarrow\mathbb{R}_{+}$. The range of $V$ is finite and it is $\{v_1,v_2,\cdots,v_K\}$ with $v_{k-1}<v_{k}$. We use $\theta\in[0,1]$ to denote a representative consumer.


In this setting, a \emph{market} $X$ is a subset of $[0,1]$. We assume that there are $n$ market segments, indexed by $i = 1, \dots, n$, where $n \geq K$ is exogenously given. A \emph{market segmentation} is a collection of markets $\mathcal{P}=\{X_1, \dots, X_n\}$ such that $X_i \cap X_j = \emptyset$ for all $i \neq j$, and $\bigcup_{i=1}^n X_i = [0,1]$. Each set $X_i$ represents the group of consumers who choose segment $i$. We say that a consumer $\theta$ is assigned to segment $i$ if $\theta \in X_i$. The producer observes the composition of each segment and can set different prices across segments.

\paragraph{Timeline.} The timeline of the extensive-form game is summarized as follows. 

\begin{enumerate}
    \item Each consumer $\theta \in [0,1]$ chooses a segment $i \in \{1, \dots, n\}$.
    \item This induces a market segmentation $\{X_1, \dots, X_n\}$, where $X_i$ is the set of consumers assigned to segment $i$. The producer observes the composition of each segment and sets a price for each segment.
    \item Each consumer makes a purchase decision.
\end{enumerate}

\paragraph{Producer's Strategies.} In stage 2, the producer offers a take-it-or-leave-it price to maximize the profit in each market. We say $p^*\in \mathbb{R}_{+}$ is an \emph{optimal price} of market $X$ if 
$$
p^*\in\arg\max_{p\geq 0} \int_{X} p\mathbb{I}\{V(\theta)\geq p\}d\theta.
$$
Multiple optimal prices might exist for a market. Let $\phi$ denote the deterministic pricing strategy of the producer, which maps every market to a non-negative price. If $\phi(X)$ is an optimal price of market $X$ for any $X\subseteq[0,1]$, the strategy $\phi$ is called \emph{rational}.

\paragraph{Consumer's Strategies.} We only consider deterministic strategies of consumers. In stage 1, each consumer $\theta\in[0,1]$ independently chooses one of $n$ market segments to join. The segment that consumer $\theta$ joins is denoted by $\sigma(\theta)\in\{1,2,\cdots,n\}$. Given strategy profile $\sigma$, we say $\mathcal{P}$ is the induced market segmentation if $X_i=\{\theta|\sigma(\theta)=i\}$. In stage 3, consumer $\theta$ makes a purchase decision if the price offered by the producer $\phi(X_{\sigma(\theta)})$ is no larger than the willingness-to-pay $V(\theta)$. 

\paragraph{Payoffs.} Let $u,\pi,$ and $w$ denote the consumer surplus, producer surplus, and total surplus, respectively. Given producer's strategy $\phi$ and consumers' strategies $\sigma$, the producer surplus is $\sum_{i=1}^{n}\phi(X_i)\int_{X_i}\mathbb{I}\{V(\theta)\geq \phi(X_i)\}d\theta$. Meanwhile, the payoff of consumer $\theta$ is denoted by $u_{\theta}(\phi,\sigma)=\max\{V(\theta)-\phi(X_{\sigma(\theta)}),0\}$. Hence, the consumer surplus is $\sum_{i=1}^{n}\int_{X_i}u_{\theta}(\phi,\sigma)d\theta$; and the social welfare is $\sum_{i=1}^{n}\int_{X_i}V(\theta)\mathbb{I}\{V(\theta)\geq \phi(X_i)\}d\theta$. 

For the aggregate market, we generically assume that the optimal uniform price $v^*$ is unique.\footnote{If multiple optimal uniform prices exist, our analysis extends by defining $v^*$ as the largest among them. This choice is natural: a higher price minimizes consumer surplus, and the uniform monopoly consumer surplus serves as the lower bound in our characterization.} We assume $v^*>v_1$ to avoid trivial analyses. Under the uniform price monopoly, we use $u^*,\pi^*$ to denote the consumer surplus and the producer surplus, respectively. Let $\overline{w}$ denote the maximum social welfare, i.e. $\int_{0}^{1}V(\theta)d\theta$. 

\paragraph{Solution Concept.} 

We now describe how we model individual deviations in the continuum framework. The key challenge is that a single consumer has zero Lebesgue measure and therefore cannot affect prices. To address this, we draw on the correspondence between discrete and continuous settings. In a market with $N$ consumers, each individual holds a market share of $\varepsilon = \frac{1}{N}$: small but non-negligible. In the continuum limit, we represent an individual consumer as a set of consumers with infinitesimal but strictly positive measure.\footnote{In the discrete model (see \autoref{Section:Discrete}), the smallest set with a positive market share is a single consumer. The continuum analogue is a measurable set whose measure can be made arbitrarily small but remains strictly positive.} A \emph{unilateral deviation} then corresponds to such a set switching from one market segment to another.

We formalize this by defining a profitable deviation for a positive-measure set of consumers. The restriction to positive measure is essential: only such a deviation can alter the composition of market segments and, consequently, prices and surplus.

\begin{definition}\label{def:dev_set}
    Given a strategy profile $(\phi,\sigma)$, a set of consumers $D\subseteq [0,1]$ is a \textbf{profitable deviation set} if (i) $D$ has positive measure, and (ii) there exists an alternative consumer strategy $\sigma'$ such that
    \begin{itemize}
        \item $\sigma'(\theta)\neq \sigma(\theta)$ if and only if $\theta\in D$, and
        \item $u_{\theta}(\phi,\sigma')-c>u_{\theta}(\phi,\sigma)$ for all $\theta\in D$,
    \end{itemize}
    where $c\geq0$ represents the cost of deviation.
\end{definition}

The deviation cost is zero ($c=0$) throughout the main analysis; we extend to positive costs in \autoref{sssection:cost}. Since we adopt a non-cooperative approach, a deviation must strictly benefit every consumer in $D$.\footnote{This does not entail collective coordination. The set $D$ consists of consumers of the same type who face identical payoffs; it merely represents the positive measure assigned to a single individual in the continuum model.} We next define a profitable unilateral deviation to capture the notion of an individual's deviation.

\begin{definition}\label{def:UDev}
    A \textbf{profitable unilateral deviation} exists if for any $\varepsilon>0$, there exists a profitable deviation set with measure smaller than $\varepsilon$.
\end{definition}

The ``for any $\varepsilon > 0$'' quantifier ensures robustness for large population sizes: in a market with $N$ consumers, each individual's share is $1/N$, so requiring stability only above some fixed threshold $\delta$ would fail to capture individual deviations once $N > 1/\delta$.

We can now define our solution concept.

\begin{definition}\label{def:SP}
    A strategy profile $(\phi,\sigma)$ induces a \textbf{strategy-proof} market segmentation $\mathcal{P}$ if (i) $\phi$ is rational, and (ii) no profitable unilateral deviation exists.
\end{definition}

\section{The Main Result and the Proof}\label{Section:Result}

We first present the main result of this paper and then provide the proof.

\begin{theorem} \label{thm:SP_surplus}
There exists a strategy profile $(\phi,\sigma)$ that induces strategy-proof segmentation \textbf{if and only if} producer surplus $\pi = \pi^{*}$ and consumer surplus $u\in[u^{*},\overline{w}-\pi^{*}]$. No positive measure of consumers is worse off compared with the uniform monopoly outcome.
\end{theorem}

Before presenting the full proof, we illustrate the result and its key intuitions using a two-type example. Consumers have either high ($v_H$) or low ($v_L$) willingness to pay, and we assume that $v_H$ is the unique optimal uniform price. Without loss of generality, the aggregate market is divided into two segments that consumers freely choose between. Let market 1 contain masses $\alpha$ and $\beta$ of high- and low-type consumers, respectively. If $v_L(\alpha+\beta)\geq \alpha v_H$, it is optimal for the producer to charge $v_L$ in market 1. Since $v_H$ is the uniform monopoly price, market 2 must contain a positive measure of high-type consumers.

Can a price of $v_L$ in market 1 be sustained? The key question is what deters a high-type consumer in market 2 from switching to market 1. When $v_L(\alpha+\beta)> \alpha v_H$, the inequality is strict, so a high-type consumer with infinitesimal measure can join market 1 without disturbing the price $v_L$---a profitable deviation. When $v_L(\alpha+\beta)= \alpha v_H$, however, any such entry tips the balance and triggers a price increase to $v_H$, eliminating the incentive to switch. This indifference condition is therefore necessary for any strategy-proof segmentation with price $v_L$ in market 1. Moreover, for a given $\beta$, it uniquely pins down the entire segmentation.

Two consequences follow. First, producer surplus equals the uniform monopoly profit, since $v_H$ remains an optimal price in market 1. Second, consumer surplus depends on $\beta$: when $\beta=0$, the outcome coincides with the uniform monopoly; when $\beta$ equals the total mass of low-type consumers, the buyer-optimal outcome is achieved. By varying $\beta$, every intermediate welfare outcome between these two extremes is attained.

Now we provide the detailed proof of our theorem.

\subsection*{Proof for Necessity}

For market $X$, let $\textbf{supp}\{X\}$ denote the \emph{support set} of the valuation distribution in market $X$. Formally, this support set comprises those valuations at which the measure of consumers exhibiting the corresponding willingness-to-pay is strictly positive.

In our illustrative example involving valuations $v_H$ and $v_L$, we demonstrate that if $v_L$ is optimal in a given market, then $v_H$ must also be optimal. This observation provides key insights into maintaining market segmentation with distinct pricing across different markets. Building on this, we generalize the condition as follows.

\paragraph{Indifference Condition.} If there exist two markets with different prices, $\phi(X_{i})<\phi(X_{j})$, then for all $v_k\in \mathbf{supp}\{X_{j}\}\cap\left(\phi(X_{i}),\phi(X_{j})\right]$, $v_k$ is also optimal in market $X_{i}$.
\bigskip

In the following, we prove that the \textbf{indifference condition} is necessary to induce a strategy-proof segmentation when the producer's strategy $\phi$ is rational.

\begin{lemma} \label{lemma:SPprice}
The strategy profile $(\phi,\sigma)$ induces a strategy-proof segmentation $\mathcal{P}=\{X_1,\cdots,X_n\}$ \textbf{only if} (i) $\phi$ is \textbf{rational}; and (ii) the \textbf{indifference condition} holds. 
\end{lemma}

\begin{proof}[Proof of \autoref{lemma:SPprice}]

Suppose $(\phi,\sigma)$ induces a strategy-proof segmentation $\mathcal{P}=\{X_{1},\cdots,X_{n}\}$. If $\phi(X_i)<\phi(X_{j})$ for markets $X_i$ and $X_j$, for any consumer in market $X_{j}$ whose value $v_k$ lies in $\left(\phi(X_i),\phi(X_j)\right]$, his payoff is zero. 

We prove this lemma by contradiction. Suppose $v_k$ is not an optimal price of market $X_i$, we define 
\begin{equation*}
    r=\phi(X_i)\int_{X_i}\mathbb{I}\{V(\theta)\geq\phi(X_i)\}d\theta - v_k\int_{X_i}\mathbb{I}\{V(\theta)\geq v_k\}d\theta>0.
\end{equation*} 
Then, for any $\varepsilon>0$, since the set $\{\theta\in X_j:V(\theta)=v_k\}$ has positive measure, we can choose a subset $D\subseteq \{\theta\in X_j:V(\theta)=v_k\}$ such that $0<\mu(D)<\min\{\varepsilon,\frac{r}{v_k-\phi(X_i)}\}$, where $\mu(D)$ is the measure of set $D$. Consider the following possible deviation
\begin{equation*}
\sigma'(\theta)=\begin{cases}
\sigma(\theta), & \theta \not\in D, \\
i, & \theta \in D.
\end{cases}
\end{equation*}
Consider the market $X_i'=X_i\cup D$ induced by strategy $\sigma'$. The difference of the revenue of pricing $\phi(X_i)$ and the revenue of pricing $v_k$ is $r+\phi(X_i)\mu(D)-v_k\mu(D)>0$. Thus, $v_k$ cannot be the optimal price of market $X_i'$. Since $D$ only contains consumers with valuation $v_k$, any $v>v_k$ cannot be the optimal price of market $X_i'$. Therefore, $\phi(X_i')<v_k$. 

This implies that deviation is profitable for all consumers in $D$, resulting in a contradiction. 
\end{proof}

Consider a strategy-proof segmentation $\mathcal{P}=\{X_{1},\cdots,X_{n}\}$ induced by strategy profile $(\phi,\sigma)$. \autoref{lemma:maxprice}, the key step for necessity, states that the maximum price across all markets ($\max\left\{\phi(X_i)\right\}_{i=1}^n$) must equal the uniform monopoly price $v^*$.

\begin{lemma}\label{lemma:maxprice}
For any strategy-proof segmentation $\mathcal{P}=\{X_1,\cdots,X_n\}$, $\max\left\{\phi(X_i)\right\}_{i=1,\textbf{supp}\{X_i\}\neq \emptyset}^n=v^*$. 
\end{lemma}

\begin{proof}[Proof of \autoref{lemma:maxprice}]

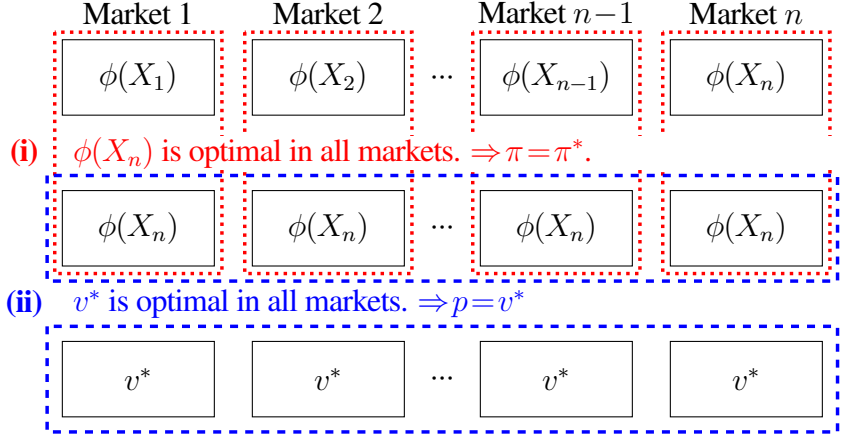
\begin{figure}[!htb]
        \centering
        \begin{tikzpicture}
        \draw (0,0) rectangle (2,-1) node at (1,0.35){Market 1}node at (1,-0.5){$\phi(X_1)$};
        \draw (2.5,0) rectangle (4.5,-1) node at (3.5,0.35){Market 2}node at (5,-0.5){$\cdots$}node at (3.5,-0.5){$\phi(X_2)$};
        \draw (5.5,0) rectangle (7.5,-1) node at (6.5,0.35){Market $n-1$}node at (6.5,-0.5){$\phi(X_{n-1})$};
        \draw (8,0) rectangle (10,-1) node at (9,0.35){Market $n$}node at (9,-0.5){$\phi(X_n)$};
        
        \draw (0,-2) rectangle (2,-3)node at (1,-2.5){$\phi(X_n)$};
        \draw (2.5,-2) rectangle (4.5,-3)node at (5,-2.5){$\cdots$}node at (3.5,-2.5){$\phi(X_n)$};
        \draw (5.5,-2) rectangle (7.5,-3)node at (6.5,-2.5){$\phi(X_n)$};
        \draw (8,-2) rectangle (10,-3)node at (9,-2.5){$\phi(X_n)$};
        
        \draw (0,-4) rectangle (2,-5)node at (1,-4.5){$v^*$};
        \draw (2.5,-4) rectangle (4.5,-5)node at (5,-4.5){$\cdots$}node at (3.5,-4.5){$v^*$};
        \draw (5.5,-4) rectangle (7.5,-5)node at (6.5,-4.5){$v^*$};
        \draw (8,-4) rectangle (10,-5)node at (9,-4.5){$v^*$};
        
        \draw[dotted,line width=1.3pt,color = red] (-0.1,0.1) rectangle (2.1,-3.1);
        \draw[dotted,line width=1.3pt,color = red] (2.4,0.1) rectangle (4.6,-3.1);
        \draw[dotted,line width=1.3pt,color = red] (5.4,0.1) rectangle (7.6,-3.1);
        \draw[dotted,line width=1.3pt,color = red] (7.9,0.1) rectangle (10.1,-3.1);
        
        \draw[dashed,line width=1.3pt,color = blue] (-0.2,-1.8) rectangle (10.2,-3.2);
        \draw[dashed,line width=1.3pt,color = blue] (-0.2,-3.8) rectangle (10.2,-5.2);
        
        \fill[color = white] (0.1,-1.28) rectangle (11,-1.72);
        \draw[color = red] node at(-0.5,-1.5){\textbf{(i)}}node at(0,-1.5)[right]{$\phi(X_{n})$ is optimal in all markets. $\Rightarrow\pi=\pi^*$.};
        \draw[color = blue] node at(-0.5,-3.5){\textbf{(ii)}}node at(0,-3.5)[right]{$v^*$ is optimal in all markets. $\Rightarrow p=v^*$};
        
        \end{tikzpicture}
        \caption{Why $\max\left\{\phi(X_i)\right\}_{i=1}^n=v^*$?}
        \label{fig:v*}
    \end{figure}

\autoref{fig:v*} graphically illustrates the main logic to prove \autoref{lemma:maxprice}. The three panels represent three different pricing strategies. In the upper panel, the producer charges $\phi(X_{i})$ in market $i$. We assume $\phi(X_{1})\leq \phi(X_{2})\leq\cdots\leq \phi(X_{n})$, which is without loss of generality by relabeling. The producer surplus, in this case, is denoted by $\pi$. In the middle panel, the producer charges a uniform price $\phi(X_{n})$ in all markets. Denote the producer surplus in this case by $\pi'$. In the lower panel, the producer charges the optimal uniform price $v^*$ in all markets and obtains the uniform monopoly profit $\pi^*$. 

We first argue that the three pricing strategies yield the same producer surplus. By \autoref{lemma:SPprice} and noticing that $\phi(X_{n})\in\textbf{supp}\{X_n\}$, $\phi(X_{n})$ should be optimal in every market. Thus, the pricing strategies in the upper and middle panels result in the same producer surplus, that is, $\pi=\pi'$. Evidently, $\pi'\leq\pi^*$, since the latter is the uniform monopoly profit and the former is the profit by charging a uniform price $\phi(X_{n})$. Moreover, from the results of the surplus triangle by \cite*{bergemann2015limits}, we know that  $\pi^*\leq\pi$. Thus, we must have $\pi=\pi'=\pi^*$.

Finally, since $\pi'=\pi^*$, $\phi(X_{n})$ must be an optimal uniform price, suggesting that $v^*=\phi(X_{n})$ by the uniqueness of the optimal uniform price.
\end{proof}

\autoref{lemma:maxprice} reveals that strategic consumer mobility prevents the producer from pricing above $v^*$. All results shown in \autoref{thm:SP_surplus} are direct implications of \autoref{lemma:maxprice}.

\begin{itemize}
    \item \textbf{Producer surplus}. Since $v^*$ is optimal in all markets within segmentation $\mathcal{P}$, the producer surplus is unchanged if charging $v^*$ to all markets, resulting $\pi=\pi^*$.
    \item \textbf{Pareto improvement}. No positive measure of consumers face a higher price than $v^*$.
    \item \textbf{Consumer surplus}. By Pareto improvement, consumer surplus is at least the uniform monopoly outcome, $u\geq u^{*}$. Meanwhile, $u\leq\overline{w}-\pi^*$ by the surplus triangle.
\end{itemize}

\subsection*{Proof for Sufficiency}

We need to construct strategy-proof segmentations to realize every $(\pi^{*},u)$ with $u\in[u^{*},\overline{w}-\pi^*]$. For one extreme case, the unsegmented market can achieve $u=u^{*}$. The other extreme case ($u=\overline{w}-\pi^*$) is realized by greedy procedures we will explain immediately. During the remaining construction process, we only use the minimum optimal pricing rule $\phi^{\min}$ that maps every market to its minimum optimal price.

To simplify our further construction, we introduce a new perspective on market segmentation. We say $\{\textbf{x}_1,\cdots,\textbf{x}_n\}$ where $\textbf{x}_i\in\mathbb{R}^{K}$ is a \emph{canonical representation} of market segmentation $\mathcal{P}=\{X_1,\cdots,X_n\}$ if $x_{i,j}=\mu(\{\theta\in X_i:V(\theta)=v_j\})$ for all $i=1,\cdots,n$ and $j=1,2,\cdots,K$. It is obvious that for any $\{\textbf{x}_1,\cdots,\textbf{x}_n\}$, there is a strategy $\sigma$ that induces a market segmentation with canonical representation $\{\textbf{x}_1,\cdots,\textbf{x}_n\}$ if and only if $\sum_{i=1}^{n}x_{i,j}=\mu(\{\theta\in [0,1]:V(\theta)=v_j\})\triangleq x_{j}^*$. Therefore, in the following construction, we focus on the canonical representation of market segmentation where $\sum_{i=1}^{n}\textbf{x}_i=\textbf{x}^*$.

The constructive approach necessarily requires non-uniform pricing, indicating that multiple optimal prices may exist in some markets (\autoref{lemma:SPprice}). Since we seek to implement the extreme case with $u=\overline{w}-\pi^*$, we enable multiple optimal prices in an extreme way: For every market within the segmentation, the producer is indifferent between charging any price inside the support, which coincides with the segmentation based on \emph{extremal market} \citep*{bergemann2015limits}.

For a set of valuations $S\subseteq \{v_1,\cdots,v_K\}$, the extremal market $\textbf{x}^{S}(a)$ with total share being $a$, is defined as 
$$
x_i^{S}(a)=\left\{
\begin{matrix}
0 & v_i\not\in S \\
a\min S\left(\frac{1}{v_i}-\frac{1}{\xi(v_i,S)}\right) & v_i\in S
\end{matrix}
\right.
$$
where $\xi(v_i,S)$ denotes the smallest element in $S$ higher than $v_i$, and $1/\xi(\max S,S)=0$. We can verify that the producer is indifferent to charging any valuation that appears in an extremal market, resulting in a profit of $a\min S$. Hence, extremal markets are sensitive to entry and helpful for constructing strategy-proof segmentations.

With these good properties in mind, we generate markets iteratively. Recall that $\textbf{supp}$ denotes the support set of a distribution. First, pack as many consumers as possible into an extremal market $\mathbf{x}^{\textbf{supp}\{\mathbf{x}^*\}}(a_1)$ until running out of consumers with some valuation. The residual market is $\mathbf{x}^{(1)}=\mathbf{x}^*-\mathbf{x}^{\textbf{supp}\{\mathbf{x}^*\}}(a_1)$. Second, pack as many consumers as possible from $\mathbf{x}^{(1)}$ into an extremal market $\mathbf{x}^{\textbf{supp}\{\mathbf{x}^{(1)}\}}(a_2)$. The residual market is thus $\mathbf{x}^{(2)}=\mathbf{x}^{(1)}-\mathbf{x}^{\textbf{supp}\{\mathbf{x}^{(1)}\}}(a_2)$. In each round $k$, we will obtain an extremal market $\mathbf{x}^{\textbf{supp}\{\mathbf{x}^{(k-1)}\}}(a_k)$ and a residual market $\mathbf{x}^{(k)}$. 
This iteration is terminated when the remaining market becomes $\mathbf{0}$. Since $|\mathbf{supp}\{\mathbf{x}^{(i+1)}\}|<|\mathbf{supp}\{\mathbf{x}^{(i)}\}|$, there is at most $K$ iterations, and the resulting segmentation is called \emph{greedy segmentation} $\mathcal{P}^\text{Greedy}(\mathbf{x}^*)=\{\mathbf{x}^\text{Greedy}_1,\cdots,\mathbf{x}^\text{Greedy}_t,\textbf{0},\cdots,\textbf{0}\}$\footnote{Suppose the iteration lasts for $t$ rounds, the remaining $n-t\geq 0$ markets are supplemented by $\textbf{0}$.}. The greedy segmentation reaches $u=\overline{w}-\pi^*$. More importantly, the greedy segmentation is strategy-proof. 

\begin{remark}\label{remark:greedy}
The greedy segmentation $\mathcal{P}^{\emph{Greedy}}(\mathbf{x}^*)$ is strategy-proof.
\end{remark}

This remark holds according to the following \autoref{lem:min_sufficent}. Combining it with \autoref{lemma:SPprice}, we can derive a necessary and sufficient condition for a market to be strategy-proof under the producer strategy $\phi^{\min}$.

\begin{lemma}\label{lem:min_sufficent}
    The strategy profile $(\phi^{\min},\sigma)$ induces a strategy-proof segmentation $\mathcal{P}=\{X_1,\cdots,X_n\}$ \textbf{if} the \textbf{indifference condition} holds.
\end{lemma}

\begin{proof}
See \autoref{app:proof-Lemma3}.
\end{proof}




For sufficiency of \autoref{thm:SP_surplus}, it remains to construct a family of strategy-proof segmentations that can achieve any intermediate $u\in(u^*,\overline{w}-\pi^*)$ based on greedy segmentation $\mathcal{P}^\text{Greedy}(\mathbf{x}^*)$. Given parameters $\alpha\in[0,1]$ and $k\in[1,t-1]$, we merge those markets that are generated later in greedy procedures to obtain a new segmentation, denoted by $\mathcal{P}^{\alpha,k}(\mathbf{x}^*)$: 
\begin{equation*}
\left\{\underbrace{\mathbf{x}_{1}^\text{Greedy},\cdots,\mathbf{x}_{k-1}^\text{Greedy},\alpha \mathbf{x}_{k}^\text{Greedy}}_\text{Extremal Markets},\underbrace{\sum_{j=k+1}^{t}\mathbf{x}_{j}^\text{Greedy}+(1-\alpha)\mathbf{x}_{k}^\text{Greedy}}_\text{Last Market},\underbrace{\textbf{0},\cdots,\textbf{0}}_\text{Supplementary Markets}\right\}.
\end{equation*}
Notably, the first $k$ markets are extremal. Meanwhile, the minimum optimal price of the last market is $v^*$, since (i) $v^*$ is optimal in every market generated by the greedy procedures, and (ii) $v^*$ is the minimum optimal price in $\mathbf{x}^\text{Greedy}_t$ (\autoref{lemma:maxprice}).\footnote{To see why the price of the last market is $v^*$, it suffices to show that (I) pricing any $v<v^*$ is strictly worse than $v^*$ and (II) pricing any $v>v^*$ is not better than $v^*$. (I) holds since $v^*$ is optimal in all markets and any $v<v^*$ is strictly worse than $v^*$ in $\mathbf{x}_t^\text{Greedy}$ since $v^*=\phi^{\min}(\mathbf{x}_t^\text{Greedy})$. (II) holds immediately from (i).} Moreover, it is not hard to show that $\mathcal{P}^{\alpha,k}(\mathbf{x}^*)$ is strategy-proof by \autoref{lem:min_sufficent}.

\begin{remark}\label{remark:constructed}
The constructed segmentation $\mathcal{P}^{\alpha,k}(\mathbf{x}^*)$ is strategy-proof.
\end{remark}

Let $\Psi$ denote the share of the last market, $\Psi=\left\|\sum_{j=k+1}^{t}\mathbf{x}_{j}^\text{Greedy}+(1-\alpha)\mathbf{x}_{k}^\text{Greedy}\right\|_1$. Given any $\Psi\in\left[\|\mathbf{x}^\text{Greedy}_t\|_1,1\right]$, there exists a unique $\mathcal{P}^{\alpha,k}(\mathbf{x}^*)$ such that the share of the last market is indeed $\Psi$. Meanwhile, the consumer surplus is continuous (and monotonic) in $\Psi$.\footnote{As $\Psi$ increases by $\text{d}\Psi$, those $\text{d}\Psi$ consumers, once in the second last market $\alpha\mathbf{x}_k^\text{Greedy}$, face price $v^*$ instead, while other consumers face the same price as before. Hence, the decline of consumer surplus is solely influenced by those changed consumers. Since changed consumers are distributed proportional to $\mathbf{x}_k^\text{Greedy}$, the loss of consumer surplus is bounded by $\text{d}\Psi\left(v^*-\phi^{\min}(\mathbf{x}_k^\text{Greedy})\right)$: Continuity holds.} Since the consumer surplus is $\overline{w}-\pi^*$ when $\Psi=\|\mathbf{x}^\text{Greedy}_t\|_1$ and $u^*$ when $\Psi=1$, all consumer surplus within $[u^*,\overline{w}-\pi^*]$ can be achieved. This ends the proof of \autoref{thm:SP_surplus}.

\section{Discussions}\label{Section:Discussion}

\subsection{Concerns for Information Leakage}

All strategy-proof segmentations constructed above share a common structure: a \emph{last market} priced at the uniform monopoly level, alongside several discount markets. This structure admits a natural interpretation. The last market serves as an \emph{anonymous market} for consumers who exercise the right to be forgotten. These consumers refuse to provide any information and therefore face the uniform monopoly price, while consumers in each discount market voluntarily reveal some information in exchange for preferential pricing.

We formalize this interpretation through a refinement of strategy-proof segmentation. In each segmentation, exactly one market is designated as the anonymous market, representing consumers who delete their accounts. In addition to the conditions in \autoref{def:SP}, we impose a \emph{no-logout} condition: any consumer who can purchase the good at a weakly lower price without logging in will do so to prevent information leakage. We restrict this condition to consumers who make a purchase because information leakage occurs primarily during transactions.

\begin{definition}[No-logout Condition]
Assume one market within the segmentation represents the \textbf{anonymous market}. For each consumer who purchases in a non-anonymous market, deviating to the anonymous market would result in a \textbf{strictly higher price}.
\end{definition}

After imposing the no-logout condition, every surviving strategy-proof segmentation contains exactly one market priced at $v^*$: the anonymous market. Moreover, this refinement is without loss. For any strategy-proof segmentation, merging all markets priced at $v^*$ into a single anonymous market yields a strategy-proof segmentation that satisfies the no-logout condition. Consequently, the set of attainable welfare outcomes is unchanged.

\subsection{Deviation Costs} \label{sssection:cost}

In this section, we assume that the deviation to another market incurs a cost of $c\geq0$ for each consumer. Alternatively, $c$ can capture the bounded rationality of consumers, in the sense that consumers cannot identify a profitable deviation unless that deviation improves by at least $c$. When $c=0$, the benchmark analysis applies.

We use $\textbf{SP}(c)$ to denote the possible set of surplus pairs $(\pi,u)$ that can be achieved by strategy-proof segmentation given cost $c$. By our main result, $\textbf{SP}(0)=\{(\pi^*,u):u\in[u^{*},\overline{w}-\pi^{*}]\}$. Moreover, the surplus triangle of \cite*{bergemann2015limits} can be interpreted as another extreme case: $\textbf{SP}(\infty)=\{(\pi,u):\pi\geq \pi^*,u\geq 0,\pi+u\leq \overline{w}\}$.

We can directly get the following proposition.

\begin{proposition}\label{prop:dev_cost_subset}
    $\textbf{\emph{SP}}(c_1)\subseteq \textbf{\emph{SP}}(c_2)$ when $0\leq c_1\leq c_2$.
\end{proposition}

The proof of \autoref{prop:dev_cost_subset} is straightforward. Consider any strategy-proof segmentation with a deviation cost of $c_1$. By definition, no unilateral deviation exists where every consumer can improve their payoff by more than $c_1$. Consequently, for any higher deviation cost $c_2\geq c_1$, the same segmentation remains strategy-proof, as the incentive for deviation is further diminished.

\autoref{prop:dev_cost_subset} implies that, when the deviation cost is strictly positive, the set of potential welfare outcomes nests our result in the costless benchmark, but is nested by the surplus triangle of \cite*{bergemann2015limits}. We use the following example to visualize this observation. 

\begin{example}\label{example-1}
    We consider an example such that there are three different reservation prices $\{1,2,3\}$ where $V(\theta)=1$ for $\theta\in[0,\frac{1}{3}]$, $V(\theta)=2$ for $\theta\in(\frac{1}{3},\frac{2}{3}]$, and $V(\theta)=3$ for $\theta\in(\frac{2}{3},1]$. We can calculate that $\pi^*=\frac{4}{3},u^*=\frac{1}{3}$, and $\overline{w}=2$.
\end{example}

In this example, since the valuations have a minimum difference of 1, the deviation cost will influence the strategy-proof segmentation only if $c\geq 1$. Therefore, by \autoref{thm:SP_surplus}, $\textbf{SP}(c)=\{(\pi,u):\pi=\frac{4}{3},u\in[\frac{1}{3},\frac{2}{3}]\}$ for $c\in[0,1)$. When $1\leq c< 2$, consumers with valuations 1 and 2 have no incentives to deviate. In this case, we only need to consider those consumers who have a valuation of $3$ and are in a market with a price 3. We can calculate that $\textbf{SP}(c)=\{(\pi,u):\frac{4}{3}\leq \pi\leq \frac{5}{3},0\leq u,u+\pi\leq 2\}$ for $c\in[1,2)$. When $c\geq 2$, no consumers have incentives to deviate, so all segmentation must be strategy-proof, resulting in $\textbf{SP}(c)=\{(\pi,u):\pi\geq \frac{4}{3},u\geq 0,\pi+u\leq 2\}$ for $c\geq2$. We plot these sets in \autoref{fig:sur_tri_subset} where $AC$ represents the surplus pairs when the deviation cost is smaller than 1, the trapezoid $EFCD$ represents the surplus pairs when the deviation cost $c\in[1,2)$, and the whole triangle $BCD$ represents the case where $c\geq 2$.

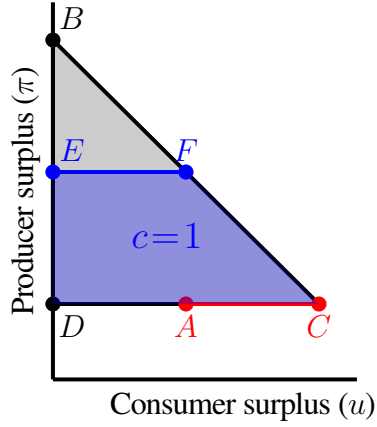
\begin{figure}[!htb]
\centering
\begin{minipage}[t]{0.45\linewidth}
\begin{tikzpicture}[scale=5]
\draw[line width=1.5pt](0,0)--(0.8,0)node[below]at(0.5,0){Consumer surplus ($u$)};
\draw[line width=1.5pt](0,0)--(0,1)node[left]at(0,0.5){\rotatebox{90}{Producer surplus ($\pi$)}};
\draw[line width=1.5pt](0,0.2)--(0,0.9)--(0.7,0.2)--(0,0.2);
\fill[fill opacity=0.2](0,0.2)--(0,0.9)--(0.7,0.2);
\fill (0,0.2) circle (0.02)node[below]at(0.05,0.2){$D$};
\fill[color = red] (0.7,0.2) circle (0.02)node[below]{$C$};
\fill[color = red] (0.35,0.2) circle (0.02)node[below]{$A$};
\fill (0,0.9) circle (0.02)node[above]at(0.05,0.9){$B$};
\draw[line width=1.5pt,color=red] (0.35,0.2) -- (0.7,0.2);

\fill[fill opacity=0.3,color=blue](0,0.2)--(0,0.55)--(0.35,0.55)--(0.7,0.2);

\fill[color = blue] (0,0.55) circle (0.02)node[above]at(0.05,0.55){$E$};

\fill[color = blue] (0.35,0.55) circle (0.02)node[above]{$F$};

\draw node[above,color = blue]at(0.3,0.32){\large{$c=1$}};

\draw[line width=1.5pt,color=blue](0,0.55)--(0.35,0.55);

\end{tikzpicture}
\end{minipage}
\caption{Welfare Consequences with Different Deviation Costs}\label{fig:sur_tri_subset}
\end{figure}

\subsection{Comparison with Subgame Perfect Equilibrium (SPE)}\label{Section:SPE}

In this section, we study subgame perfect equilibrium and show that strategy-proof market segmentation, as a solution concept, is a refinement of subgame perfect equilibrium. 

\begin{definition}\label{def:SPE}
    A strategy profile $(\phi,\sigma)$ is a \textbf{subgame perfect equilibrium} if (i) $\phi$ is rational; (ii) for any $\theta\in[0,1],k\in\{1,2,\cdots,n\}$, 
    $u_{\theta}(\phi,\sigma(\theta),\sigma_{-\theta})\geq u_{\theta}(\phi,k,\sigma_{-\theta})$ where $\sigma_{-\theta}$ is the strategy profile of all consumers except for consumer $\theta$.
\end{definition}

The distinction between strategy-proof segmentation and SPE lies in their respective approaches to modeling individual deviations. In the context of strategy-proof segmentation, an individual's deviation constitutes a minimal deviation set, which consequently alters the canonical representation of market segmentation. In contrast, under the framework of SPE, an individual deviation does not modify the canonical representation of market segmentation. Yet, it should be noted that when multiple optimal prices exist within a market, the prevailing price may vary before and after a consumer's participation in one market.

In the following, we show that any strategy-proof segmentation $\mathcal{P}$ can be induced by a SPE, but not vice versa. Thus, strategy-proof segmentation can be regarded as a refinement of the SPE. 

\begin{proposition}\label{prop:SP_subset_SPE}
    For any strategy-proof market segmentation $\mathcal{P}$, there exists a subgame perfect equilibrium $(\phi,\sigma)$ that induces $\mathcal{P}$.
\end{proposition}

\begin{proof}[Proof of \autoref{prop:SP_subset_SPE}]
Given a strategy-proof segmentation $\mathcal{P}=(X_1,\cdots,X_n)$, we will construct a SPE $(\sigma,\phi)$ such that the induced segmentation is $\mathcal{P}$. 

Consumers' strategy $\sigma$ satisfies that $\sigma(\theta)=k$ where $\theta\in X_k$. The producer's strategy $\phi$ is defined as follows. First, $\phi(X_i)\triangleq\phi^{\min}(X_i)$ for all $i$. Second, for a consumer $\theta\notin X_i$ with $V(\theta)=v_k$, we have $\phi(X_i\cup\{\theta\})\triangleq\lim_{\varepsilon\rightarrow 0^+}\phi^{\min}(\mathbf{x}_i+\varepsilon \mathbf{e}_{k})$ where $\phi^{\min}(\mathbf{x})$ is the minimum optimal price of the market whose canonical representation is $\mathbf{x}$ and $\mathbf{e}_k$ is the $K$-dimensional vector that only the $k$-th component is non-zero and is equal to one. Third, for any other market $X$, $\phi(X)$ is any one of the optimal prices of market $X$. 

Since $\lim_{\varepsilon\rightarrow 0^+}\phi^{\min}(\mathbf{x}_i+\varepsilon \mathbf{e}_{k})$ is an optimal price of market $X_i$, $\phi$ is rational. Since $\phi$ takes an individual's deviation as a deviation set with infinitesimal size, once $\mathcal{P}$ is a strategy-proof segmentation, for any consumer $\theta$, it cannot deviate to another market to obtain a higher payoff. 
\end{proof}

However, a SPE might induce a market segmentation that is not strategy-proof. Using \autoref{example-1}, we find a SPE under which the producer surplus is greater than $\pi^*$, as follows. 

Consider the segmentation $\mathcal{P}=(X_1,X_2,\mathbf{0})$ with $X_1=[0,\frac{1}{3}]\cup(\frac{5}{6},1]$ and $X_2=(\frac{1}{3},\frac{5}{6}]$. The producer's strategy is $\phi(X_1)=1,\phi(X_2)=2$ and $\phi(X)$ is the maximum optimal price of $X$ for markets $X\subseteq[0,1]$ other than $X_1,X_2$. Clearly, we can verify that both $1$ and $3$ are the optimal prices of market $X_1$, while $2$ is the unique optimal price of market $X_2$. By \autoref{lemma:SPprice}, $\mathcal{P}$ is not strategy-proof. However, no consumers in market $X_2$ have an incentive to deviate to market $X_1$ since this change of $X_1$ will lead the price in market $X_1$ to be $3$. Therefore, under the strategy $\phi$, the segmentation survives under SPE. Now, the producer surplus is $\frac{3}{2}>\pi^*=\frac{4}{3}$.

\begin{remark}
The first-degree price discrimination outcome cannot be sustained in any SPE $(\phi,\sigma)$. Hence, the set of welfare outcomes attainable under SPE is a strict subset of the full surplus triangle of \cite*{bergemann2015limits}.
\end{remark}

\subsection{Discrete Consumers}\label{Section:Discrete}

Our baseline model represents the consumer population as a continuum. In this section, we consider a finite-consumer counterpart that directly captures consumers who are both numerous and individually non-negligible. The finite model introduces a natural discretization error: each consumer has mass $1/N$, so a unilateral deviation shifts a market's composition by exactly one grid point. We make this approximation error explicit and show that the set of equilibrium welfare outcomes in the discrete model converges to our baseline characterization as the number of consumers grows large.

Suppose there are $N$ consumers. Consumer $\ell$ has valuation $V_\ell\in V=\{v_1,\cdots,v_K\}$, where $0<v_1<\cdots<v_K$. Each consumer has a market share $1/N$. A market is represented by a vector $\mathbf{x}=(x_1,\cdots,x_K)\in \frac{1}{N}\mathbb{Z}_{+}^{K}$, where $x_k$ is the share of consumers with valuation $v_k$ in this market. Let $\mathbf{e}_k$ be the $k$-th unit vector. The aggregate market is denoted by $\mathbf{x}^{*,N}$, with $\sum_{k=1}^{K}x_k^{*,N}=1$.

For a nonempty market $\mathbf{x}$ and a price $p\in V$, define $R_{\mathbf{x}}(p)=p\sum_{k:v_k\geq p}x_k$. Because demands are constant between two adjacent valuations, an optimal price can be selected from $V$. Let $\Phi(\mathbf{x})=\arg\max_{p\in V}R_{\mathbf{x}}(p)$ and $\phi^{\min}(\mathbf{x})=\min\Phi(\mathbf{x})$. For the aggregate market, let
\[
\pi_N^*=\max_{p\in V}R_{\mathbf{x}^{*,N}}(p),\qquad
v_N^*\in\arg\max_{p\in V}R_{\mathbf{x}^{*,N}}(p)
\]
denote the uniform monopoly profit and a uniform monopoly price. When the maximizer is unique, $v_N^*$ denotes this unique maximizer. Let $u_N^*$ be the consumer surplus under the uniform monopoly price, and let $\overline{w}_N=\sum_{k=1}^{K}v_kx_k^{*,N}$.

A market segmentation in the finite economy is a finite collection $\{\mathbf{x}_1,\cdots,\mathbf{x}_m\}$ of nonempty markets such that $\sum_{i=1}^{m}\mathbf{x}_i=\mathbf{x}^{*,N}$. Empty markets can be ignored because a consumer entering an empty market alone faces her own valuation as the optimal price and obtains zero surplus. A rational pricing strategy $\phi$ satisfies $\phi(\mathbf{x})\in\Phi(\mathbf{x})$ for every nonempty market $\mathbf{x}$. We say the induced segmentation is \emph{strategy-proof} if no consumer can obtain a strictly higher payoff by moving alone to another market, taking into account the new optimal price in the market she enters.

For a rational strategy profile, write
\[
\pi_N=\sum_{i=1}^{m}R_{\mathbf{x}_i}(\phi(\mathbf{x}_i)),
\qquad
u_N=\sum_{i=1}^{m}\sum_{k=1}^{K}\max\{v_k-\phi(\mathbf{x}_i),0\}x_{i,k}.
\]

\begin{theorem}\label{thm:bound_discrete}
Suppose $N$ is large enough and the uniform monopoly price $v_N^*$ is unique. 

\noindent (i) Suppose $\{\mathbf{x}_1,\cdots,\mathbf{x}_m\}$ is strategy-proof, then
\[
\pi_N^*\leq \pi_N\leq \pi_N^*+\frac{m(v_K-v_1)}{N}.
\]
Moreover, every consumer is weakly better off than under the uniform monopoly outcome, and
\[
u_N^*\leq u_N\leq \overline{w}_N-\pi_N^* .
\]

\noindent (ii) Suppose $\mathbf{x}^{*,N}\rightarrow \mathbf{x}^*$. Every point in $\left\{(\pi,u):\pi=\pi^*,\ u\in[u^*,\overline{w}-\pi^*]\right\} $
is the limit of strategy-proof finite-economy outcomes: there is a sequence of strategy-proof segmentations for $\mathbf{x}^{*,N}$ whose welfare outcomes converge to that point.
\end{theorem}

\begin{proof}
See \autoref{app:proof-Theorem2}.
\end{proof}

\autoref{thm:bound_discrete} gives the finite-consumer analogue of \autoref{thm:SP_surplus}. Part (i) shows that any finite strategy-proof outcome lies within the discretization surplus $m(v_K-v_1)/N$ of the continuum necessity result, so the error vanishes when the number of markets is bounded. Part (ii) gives the limiting converse. Fix a continuum construction from the proof of \autoref{thm:SP_surplus}--the greedy segmentation or a merged segmentation $\mathcal{P}^{\alpha,k}(\mathbf{x}^*)$. If its canonical coordinates lie on the $1/N$ grid and aggregate to $\mathbf{x}^{*,N}$, the same representation is strategy-proof under $\phi^{\min}$. Otherwise, part (ii) only requires convergence, which is obtained by grid-admissible approximations along suitable values of $N$. Exact attainability at a fixed $N$ is therefore only an arithmetic issue.

\section{Concluding Remarks}

In this paper, we introduce a solution concept called strategy-proof market segmentation, designed to capture the fact that each consumer remains non-negligible when moving between segments, even in a large population. Two natural alternatives for modeling decentralized segmentation in a continuum economy are either too restrictive or too permissive: if pricing depends only on the type distribution, individual deviations have no price impact and only the monopoly outcome survives; under subgame-perfect equilibrium, even measure-zero deviations can alter prices, yielding a strictly larger set of welfare outcomes. Strategy-proof segmentation occupies the middle ground. Under any strategy-proof segmentation, producer surplus remains at the uniform monopoly level, consumer surplus is at least the uniform monopoly level, and some segmentations reach the efficient frontier. We justify this solution concept by showing that, in a discrete-consumer model, the set of Nash equilibrium welfare outcomes converges to precisely our characterization as the number of consumers grows large.

In an earlier working paper version, we also study group deviations, in which consumers from multiple markets and of different types simultaneously switch segments, provided that every member of the deviating group benefits strictly. Some strategy-proof segmentations do not survive this stronger stability requirement, but the set of attainable welfare outcomes remains unchanged. See \cite*{workingpaperversion} for details.

Our solution concept does face practical limitations. First, we assume a fully rational producer who can accurately infer the value distribution of each market segment and whose pricing algorithm responds instantaneously to even minor distributional changes. Second, we rely on the assumption that consumer valuations take finitely many values. This restriction is not merely technical: the indifference condition at the heart of our characterization requires that a small perturbation in a market's composition can discretely shift the optimal price. A finite type space guarantees discrete pricing, so that any entry by a consumer triggers a price jump. If valuations were instead drawn from a continuum, the revenue function would vary continuously in the type distribution, potentially eliminating this mechanism. Extending our results to continuous distributions remains an open question, though the discrete-consumer convergence result in \autoref{Section:Discrete} suggests that the finite-type framework is a natural modeling choice.

Technological advances are making these assumptions increasingly plausible. Modern Internet platforms can aggregate consumer data to estimate value distributions in real time, and their substantial computing power enables near-instantaneous algorithmic pricing responses. Platforms such as WeChat, which connects sellers and buyers through mini-programs, can gather and disseminate market information to facilitate efficient segmentation.\footnote{Tencent, the largest digital enterprise in China and the largest game company worldwide, derives the majority of its revenue from games and entertainment. It is therefore reasonable to view WeChat, Tencent's instant messaging service, as operating to benefit both sellers and buyers on its platform.} Large device manufacturers (e.g., Apple and Samsung), data brokers (e.g., Acxiom and Bloomberg), and data intermediaries (e.g., MiData, Salus Coop) could play a similar role.

\appendix

\section{Proof of \autoref{lem:min_sufficent}}\label{app:proof-Lemma3}

Assume that the strategy profile $(\phi^{\min},\sigma)$ satisfies the indifference condition. We prove that the market segmentation must be strategy-proof. To establish this result, we employ a proof by contradiction.

Suppose there exists a market segmentation $\mathcal{P}=\{X_1,\cdots,X_n\}$ that satisfies the indifference condition, yet admits a profitable unilateral deviation. Without loss of generality, we assume $\phi(X_{1})\leq \phi(X_{2})\leq\cdots\leq \phi(X_{n})$. Let $r$ denote the minimum revenue difference between using the optimal price and using a non-optimal price in any market segment. For any $\varepsilon\in\left(0,\frac{r}{v_K-v_1}\right)$, there must exist a deviation set $D(\varepsilon)$ with $\mu(D(\varepsilon))\in(0,\varepsilon)$ and corresponding strategy profile $\sigma'$ satisfying the conditions listed in \autoref{def:dev_set}. We assume that $X_1',\cdots,X_n'$ is the market segmentation induced by $\sigma'$. 

We first claim that $\phi^{\min}(X_i')$ must be an optimal price of market $X_i$. If not, pricing $\phi^{\min}(X_i)$ in market $X_i'$ must yield higher profit than pricing $\phi^{\min}(X_i')$ by at least $r-\mu(D)(\phi^{\min}(X_i')-\phi^{\min}(X_i))>0$. This implies that $\phi^{\min}(X_i')$ is not an optimal price of market $X_i'$, resulting in a contradiction. Since $\phi^{\min}(X_i')$ is an optimal price of $X_i$ and $\phi^{\min}(X_i)$ is the minimum optimal price, we have $\phi^{\min}(X_i) \leq \phi^{\min}(X_i')$. 

Let $L_i=D(\varepsilon)\cap X_i$ denote the set of consumers leaving $X_i$ and $J_i=D(\varepsilon)\cap X_i'$ denote the set of consumers joining this market. Thus, $X_i'=(X_i\setminus L_i)\cup J_i$. Consider a market with the smallest index $I$ such that $\mu(J_I)>0$. Since the price after deviation (i.e., $\phi^{\min}(X_I')$) will not decrease (i.e., $\phi^{\min}(X_I')\geq \phi^{\min}(X_I)$) and $I$ is the smallest index with $\mu(J_I)>0$, we must have $\mu(L_I)=0$; otherwise, consumers in $L_I$ must face a weakly higher price than $\phi^{\min}(X_I)$. 

Let $v_k=\min\textbf{supp}\{J_I\}$ and we consider one consumer $\theta$ with valuation $v_k$ coming from $X_j$, namely $\sigma(\theta)=j$. We must have $v_k>\phi^{\min}(X_I')\geq\phi^{\min}(X_I)$ to ensure that the consumer $\theta$ obtains a strictly positive payoff after the deviation. If $v_k\leq \phi^{\min}(X_j)$, then by the indifference condition, $v_k$ is an optimal price of $X_I$; otherwise (i.e., $v_k>\phi^{\min}(X_j)$), $\phi^{\min}(X_j)$ is an optimal price of $X_I$ and $\phi^{\min}(X_I')<\phi^{\min}(X_j)$ to ensure a profitable deviation. Thus, we know that $p=\min\{\phi^{\min}(X_j),v_k\}$ must be an optimal price of $X_I$ and it is strictly larger than $\phi^{\min}(X_I')$. Finally, a contradiction arises,
\begin{eqnarray*}
0&\geq&p\int_{X_I'}\mathbb{I}\{V(\theta)\geq p\}d\theta - \phi^{\min}(X_I')\int_{X_I'}\mathbb{I}\{V(\theta)\geq \phi^{\min}(X_I')\}d\theta \\
&=&(p-\phi^{\min}(X_I'))\mu(J_I)>0,
\end{eqnarray*}
where the former inequality is valid because $\phi^{\min}(X'_I)$ is an optimal price of market $X_I'$; meanwhile, the latter inequality holds because $p>\phi^{\min}(X_I')$ and $\mu(J_I)>0$. As for the equation, let $\Pi$ denote the revenue of pricing $p$ or $\phi^{\min}(X_I')$ in market $X_I$. Then, the revenue of pricing $p$ in market $X_I'$ is $\Pi+p\mu(J_I)$ while the revenue of pricing $\phi^{\min}(X_I')$ in market $X_I'$ is $\Pi+\phi^{\min}(X_I')\mu(J_I)$.

\section{Proof of \autoref{thm:bound_discrete}}\label{app:proof-Theorem2}

The lower bound is immediate. Since the producer chooses an optimal price in each market, $R_{\mathbf{x}_i}(\phi(\mathbf{x}_i))\geq R_{\mathbf{x}_i}(v_N^*)$ for every $i$. Summing over all markets gives $\pi_N\geq R_{\mathbf{x}^{*,N}}(v_N^*)=\pi_N^*$.

We now prove the upper bound. Without loss of generality, order the nonempty markets so that
\[
\phi(\mathbf{x}_1)\leq \phi(\mathbf{x}_2)\leq \cdots \leq \phi(\mathbf{x}_m).
\]
Let $p_m=\phi(\mathbf{x}_m)$. Since $p_m$ is an optimal price in the nonempty market $\mathbf{x}_m$, there is at least one consumer with valuation $p_m$ in market $m$; otherwise the producer could increase the price to the next valuation in the support without reducing demand. Let $h$ be such that $p_m=v_h$.

Consider such a consumer with valuation $p_m$. For any $i<m$, if she deviates to market $i$, the entered market becomes $\mathbf{x}_i+\frac{1}{N}\mathbf{e}_{h}$. To prevent this deviation from being profitable, the new price in market $i$ must be at least $p_m$, because the consumer obtains zero surplus in market $m$. It also cannot exceed $p_m$: if the new price were strictly above $p_m$, the entrant would not buy, and the old optimal price $\phi(\mathbf{x}_i)$ would yield strictly higher revenue after the entrant arrives. Hence $p_m$ is an optimal price of $\mathbf{x}_i+\frac{1}{N}\mathbf{e}_{h}$. For $i=m$, the inequality below is trivial.

Therefore, $R_{\mathbf{x}_i}(p_m)+\frac{p_m}{N}
\geq
R_{\mathbf{x}_i}(\phi(\mathbf{x}_i))+\frac{\phi(\mathbf{x}_i)}{N}$,
or equivalently,
\[
R_{\mathbf{x}_i}(\phi(\mathbf{x}_i))
\leq
R_{\mathbf{x}_i}(p_m)+\frac{p_m-\phi(\mathbf{x}_i)}{N}
\leq
R_{\mathbf{x}_i}(p_m)+\frac{v_K-v_1}{N}.
\]
Summing this inequality over $i=1,\cdots,m$, we obtain
\[
\pi_N
=\sum_{i=1}^{m}R_{\mathbf{x}_i}(\phi(\mathbf{x}_i))
\leq
\sum_{i=1}^{m}R_{\mathbf{x}_i}(p_m)+\frac{m(v_K-v_1)}{N}
=
R_{\mathbf{x}^{*,N}}(p_m)+\frac{m(v_K-v_1)}{N}
\leq
\pi_N^*+\frac{m(v_K-v_1)}{N}.
\]

The preceding upper-bound argument also implies
\[
R_{\mathbf{x}^{*,N}}(p_m)
\geq
\pi_N-\frac{m(v_K-v_1)}{N}
\geq
\pi_N^*-\frac{m(v_K-v_1)}{N}.
\]
When $N$ is large enough, we have
\[
\Gamma_N=\pi_N^*-\max_{p\in V\setminus\{v_N^*\}}R_{\mathbf{x}^{*,N}}(p)>\frac{m(v_K-v_1)}{N}.
\]
If $p_m\neq v_N^*$, then by definition of $\Gamma_N$,
\[
R_{\mathbf{x}^{*,N}}(p_m)\leq \pi_N^*-\Gamma_N
<
\pi_N^*-\frac{m(v_K-v_1)}{N},
\]
a contradiction. Hence $p_m=v_N^*$, so all market prices are weakly below $v_N^*$. Since every consumer faces a price weakly below the uniform monopoly price, each consumer's payoff is weakly higher than under uniform monopoly pricing. Hence $u_N\geq u_N^*$. The upper bound follows from $\pi_N\geq \pi_N^*$ and $u_N+\pi_N\leq\overline{w}_N$. 

For (ii), the convergence statement follows from the argument following \autoref{thm:bound_discrete}.

\bibliography{segmentation}
\bibliographystyle{aer}

\end{document}